\documentclass[runningheads,a4paper]{llncs}

\usepackage{amssymb}
\usepackage{amsmath}
\usepackage{url}
\usepackage{graphicx}
\usepackage{epsfig}
\usepackage{tikz}
\usepackage[normalem]{ulem}
\usepackage[inference]{semantic}

\definecolor{blue}{rgb}{0,0,1}

\newcommand{\ie}    {i.e., }

\newcommand{\wrt}   {w.r.t.\ }

\def\L{\mathcal{L}}

\def\T{\mathcal{T}}

\newcommand{\defeq}{\stackrel{\mbox{{\tiny\rm df}}}{=}}

\newcommand{\channel} {\mbox{{\sc Ch}}}
\newcommand{\var} {\mbox{{\tt Vars}}}
\newcommand{\CLines} {\mbox{\tt CLines}}
\newcommand{\hosts} {\mbox{\tt Hosts}}
\newcommand{\instances} {\mbox{\tt VMs}}
\newcommand{\addr}   {\mbox{\tt CAddr}}

\newcommand{\INT} {\mbox{\sc Ints}}

\newcommand{\true}   {\mbox{\tt true}}
\newcommand{\false}   {\mbox{\tt false}}

\newcommand{\iread}   {\mbox{\tt cread}}

\newcommand{\sleep}   {\mbox{\tt SLEEP}}

\newcommand{\move}   {\mathtt{MOVE}}
\newcommand{\istop}   {\mathtt{STOP}}
\newcommand{\iskip}   {\mathtt{SKIP}}

\newcommand{\seccloud}   {\mbox{{\tt CSP$_{4v}$ }}}

\newcommand{\tsub} {\mbox{{\sc TSub }}}
\newcommand{\tstop} {\mbox{{\sc TStop }}}
\newcommand{\tskip} {\mbox{{\sc TSkip }}}
\newcommand{\tass} {\mbox{{\sc TAssign }}}
\newcommand{\tbranch} {\mbox{{\sc TBranch }}}
\newcommand{\tloop} {\mbox{{\sc TLoop }}}

\newcommand{\tmove} {\mbox{{\sc TMove }}}
\newcommand{\tseq} {\mbox{{\sc TSeq }}}

\newcommand{\tsend} {\mbox{{\sc TSend }}}
\newcommand{\trecv} {\mbox{{\sc TRecv }}}

\newcommand{\tpar} {\mbox{{\sc TPar }}}

\urldef{\mailsa}\path|c.mu@tees.ac.uk|
\newcommand{\keywords}[1]{\par\addvspace\baselineskip
\noindent\keywordname\enspace\ignorespaces#1}

\begin{document}
\mainmatter

\title{Analysing Flow Security Properties\\ 
in Virtualised Computing Systems}

\titlerunning{Analysing Flow Security Properties in Virtualised Computing Systems}%

\author{Chunyan Mu}
\authorrunning{C. Mu}
\institute{
   Department of Computer Science, 
   Teesside University, U.K. \\
   \mailsa
}

\maketitle
 
\begin{abstract}
This paper studies the problem of reasoning about flow security properties in virtualised computing networks with mobility from perspective of formal language. We propose a distributed process algebra \seccloud with security labelled processes for the purpose of formal modelling of virtualised computing systems. Specifically, information leakage can come from observations on process executions, communications and from cache side channels in the virtualised environment. We describe a cache flow policy to identify such flows. A type system of the language is presented to enforce the flow policy and control the leakage introduced by observing behaviours of communicating processes and behaviours of virtual machine (VM) instances during accessing shared memory cache.
\end{abstract}


\section{Introduction}
\label{sec:intro}

Cloud computing has attracted interests in both the scientific and the industrial computing communities
due to its ability to provide flexible, configurable, and cost-effective computing resources delivered
over the internet.
In cloud systems, computing resources are shared among multiple clients.
This is achieved by virtualisation in which a collection of virtual machines (VMs) are running upon the same platform under the management of the hypervisor. 
Virtualisation 
allows computing service providers to maximise the utilisation of the devices and minimise the costs by creating multiple VMs over a shared physical infrastructure.
However, such services can also bring additional channel threat,
and introduce information leakage between unrelated entities 
during the procedure of \emph{resource sharing} and \emph{communications} 
through unintended covert channels.
To address this concern, this paper proposes to develop formal approaches to specifying, 
modelling and analysing flow security properties
in virtualised computing networks, 
which are the key underpin of contemporary society such as cloud computing.

Specifically, information leakage can come from observations on 
both program executions and cache usage in the virtualised environment. 
On the one hand,
for processes running upon a particular VM instance, 
consider the processes are communication channels,
sensitive inputs can be partially induced by observing public outputs of 
the processes regarding choices of public inputs.
On the other hand,
shared caches enable competing VM instances to extract sensitive information from each other.
CPU cache, a small section of memory built in the CPU for fast memory access, is one of the highest-rate measurable resources shared by multiple processes~\cite{WuXW12}.
Therefore cache-based side-channel attack becomes one of the major attack on VMs and receives the most attention in the cloud environment~\cite{RistenpartTSS09}.
Consider a malicious VM repeatedly access shared cache memory to perform cache-based side-channel attacks.
Such attacks allow one virtual machine to effectively steal secrets of another hosted machine in the same cloud environment.
More precisely, 
for processes running upon different VM instances, 
cache usage can be considered as a communication channel.
Consider the cache lines accessed by the victim instance and by the malicious instance 
as high level input and low level input respectively,
by observing the usage (such as time) of victim cache lines (low output),
the malicious observer can learn some information of the victim instance.
In summary, this paper considers distributed virtualised computing environment 
where the attacker VM steals the information from the target one 
by observing executions of victim processes, 
and by probing and measuring the usage (timing) of the shared cache.

In particular, we develop an approach from the software language-based level 
to enforce the applications to access shared cache and 
to bring interferences in a predictable way. 
As a result, we aim to prevent the leakage introduced by 
such cache timing channel and the interference between security objects caused by executions. 
First, we propose a CSP-like language for modelling communicating processes 
running upon VMs with mobility in the computing environment. 
Second, we formalise a cache flow policy to specify the security condition 
regarding the threat model we focus on.
Finally, we describe a type system of the language 
to enforce the flow policy and control the leakage 
introduced by observing the system behaviours.
More specifically, we give an identifying label to each VM instance,
and partition the cache into a set of page-sized blocks, 
each block in pages is mapped to a set of cache lines and 
is dynamically granted the label of its owner (the VM instance it is allocated to). 
The label of the block will be revoked when its owner terminates.
In addition, programming variables and cache lines are assigned security labels 
to denote the security levels of the data they store.
VM instances and hosts are assigned to different categories 
for the purpose of controlling information flow 
introduced by processes and VM instances movements.
When a guest VM is launched,
the VM manager allocates cache pages to it regarding its requests. 
During the running time of the VM process, 
operation of cache lines is allowed only if it satisfies the specified flow policy.
Processes are allowed to move from one VM to another,
and VM instances are allowed to move from one host to another.
However, movements from a lower-ordered category carrier to a higher one are not allowed.
Furthermore, we only allow the communicating process access the cache within a certain time
in order to prevent timing leaks by observing the usage of the cache.
These regulations are enforced by the semantics and the typing rules we formalised. 
When the guest VM terminates, the relevant cache pages are initialised, 
and released to the VM manager to prevent flushing of the pages by a malicious VM.
%

\paragraph{Outline.}
This paper is organised as follows. 
Section~\ref{sec:lang} describes our language for modelling basic virualised computing networks.
Section~\ref{sec:attacker} describes the threat model we focus on.
Section~\ref{sec:policy} specifies cache flow policy which the processes should satisfy in our model.
Section~\ref{sec:type} presents a type system to enforce the cache flow policy.
Section~\ref{sec:related} briefly reviews literature in the related areas.
\section{The Modelling Language \seccloud}
\label{sec:lang}
This section presents a dialect of communicating sequential processes (CSP)~\cite{Hoare78} language 
\seccloud for formal modelling of and reasoning about virtualised computing network systems considered in this paper. 
Such an environment can be considered as a distribute computing system, 
in which a group of inter-connected and virtualised computers 
are dynamically allocated for serving.
Processes can move from one VM to another and communicate to each other 
via sending/receiving messages.
Applications and data are stored and processed in the network 
but can be accessed from any location using a client.
It is natural to specify and describe the system as a set of communicating processes
in a network with consideration of resource sharing 
in a predictable way.
A CSP-like language is therefore a good choice for 
modelling of such systems.

\subsection{Terminology and notation}
We consider the infrastructure consists of a set of \emph{virtual private networks} (VPNs)
upon which a set of \emph{virtual machines} (VMs) can communicate with each other.
A VPN may include one or more VMs, 
and the location of VMs can be viewed as a node (host) of the VPN it belongs to.
An \emph{instance} is a VM 
upon which a number of \emph{processes} locate and run.
Let $I$ denote a set of instances, $I = \{I, I_1, \dots, I_n\}$.                                                                                                                                                          
Processes ($P, P_1, \dots$) can be constructed from a set of atomic actions called 
\emph{events} ($e, e_1, \dots$),
or be composed using operators to create more complex behaviours.
The full set of actions that a system may perform is called the \emph{alphabet} ($\Sigma$).
The operators are required to obey algebraic laws which can be used for formal reasoning.
The interactions carrying data values between processes take place through ``channels''.
From an information theory point of view, 
a storage device such as cache which can be received from (reading) and sent to (writing)
is also a kind of communication channel.
%
CPU data caches locate between the processor cores and the main memory.
We assume the clients (including the attacker) know the map between memory locations and cache sets,
so we omit the details of the mapping and focus on cache organisation and operation here.
Cache can be viewed as a set of cache lines:
$\CLines = \{l_i | 0 \le i \le n\}$ where $n$ denotes the size of the cache.
%

In addition, in order to encode the desired features of the language for flow secure 
virtualised computing systems, 
we assign security labels into variables and cache lines (channels),
and allocate cache lines into instances via mappings: 
\[ 
\tau_{v}: \var \mapsto_{\tau_v} \L, \
\tau_{c}: \CLines \mapsto_{\tau_c} \L, \
\alpha_{c}: \CLines \mapsto_{\alpha_c} I 
\]
where $\L$ denotes a security lattice. 
We write $\tau$ in stead $\tau_{v}$ and $\tau_{c}$ in general 
for cases without introducing any confusion.
Furthermore, we consider VM hosts are assigned to different categories $\Omega_H$, 
with an ordering of subset relations:
\[\beta_h: \hosts \mapsto_{\beta_h} \Omega_H.\]
For instance, $h_1$ is assigned to category: $\Omega_1=\{\texttt{student}\}$,
$h_2$ is assigned to category: $\Omega_2=\{\texttt{student, staff}\}$,
and thus $h_1 \sqsubseteq h_2$ since $\Omega_1 \subseteq \Omega_2$.
Similarly, we also assign VM instances into different sub-categories $\Omega_I$,
with an ordering of subset relations:
\[\beta_i: \instances \mapsto_{\beta_i} \Omega_I.\]
For instance, $P$ running upon VM $i_1$ belonging to sub-category: 
{$\omega_1 = \{\texttt{UG-1}\} \subseteq \texttt{student}$},
$Q$ running upon VM $i_2$ belonging to sub-category:
{$\omega_2 = \{ \texttt{UG-1, UG-2} \} \subseteq \texttt{student}$}, 
and $\omega_1 \subseteq \omega_2$, so $i_1 \sqsubseteq i_2$.

\subsection{Syntax}
Table~\ref{tbl:syntax} presents the syntax of the language \seccloud.
\begin{table}
\begin{center}
\begin{tabular}{ll}
\hline 
Expressions & ~~~ $\exp ::= \var \mid \channel \mid \INT \mid  \exp \oplus \exp$ 
	    \scriptsize{$~~~ (\oplus \in \{+,-,*,/\})$}
\\ 
Boolean expresssions & ~~~ $b ::= \true \mid \neg b \mid b \land b \mid \exp \bowtie \exp$ 
	    \scriptsize{$~~~ (\bowtie \in \{>,\ge, <, \le, ==\})$}
\\ 
Process & ~~~ 
$x := \exp \mid \istop \mid \iskip \mid \sleep(\Delta t) \mid$ 
\\
&~~~ ~~~ ~~~
$\move_P(i)\mid P;Q \mid P \lhd b \rhd Q  \mid b \rhd (P)^* ~|$
\\
&~~~ ~~~ ~~~
$ a!w \mid a?x \mid  P \Vert Q$
\\
VM instances  & ~~~ $I ::= [\![P]\!] \mid \move_I(h) \mid I \Vert I'$ %
\\
Hosts & ~~~ $H ::= i:[\![ I ]\!].M_I  \mid H \Vert H'$
\\ 
VPN networks & ~~~ $G ::= h:[\![H]\!] \mid G \Vert G'$
\\
\hline
\end{tabular}
\end{center}
\caption{Syntax of \seccloud}
\label{tbl:syntax}
\end{table}
Expressions can be variables, channels, integers and arithmetic operations (denoted by $\oplus$) between expressions.

%
Operator $x:=\exp$ assigns the value of $\exp$ to process variable $x$.
Action operator $\istop$ denotes the inactive processes that does nothing 
and indicates a failure to terminate,
and delayable operator $\sleep(\Delta t)$ allows the process to do nothing and wait for $\Delta t$ time units.
Moving operator $\move_P (i)$ allows a process $P$ to be able to move from current VM to another one $i$. 
We require that VM processes running upon on a VM instance 
with lower (category) order are not allowed to move to a VM instance with a higher (category) order.
Operator $P;Q$ denotes the sequential composition of processes $P$ and $Q$.
Branch operator $P \lhd b \rhd Q$ defines if the boolean expression $b$ is true
then behaves like $P$ otherwise behaves like $Q$.
Sending operator $a!(w)$ will output a value in expression $w$ over channel $a$
to an agent, and receiving operator $a?x$ allows us to input a data value during an interaction 
over channel $a$ and write it into variable $x$ of an agent.
$P \Vert Q$ denotes the synchronous parallel composition of processes $P$ and $Q$.
Loop operator $b \rhd (P)^*$ denotes the loop operation of process $P$ while $b$ is true. 

An instance $I$ is a virtual machine (VM) hosted on a network infrastructure.
Operator $ [\![P]\!]$ defines the VM upon which process $P$ runs. 
Operator $\move_I (h)$ allows VM instances (with all processes running on it) to migrate
from current host machine to another host $h$ to keep the instance running even when an event, 
such as infrastructure upgrade or hardware failure, occurs.
Similarly to the process movement, 
we require that VM instances running upon on a host
with lower (category) order are not allowed to move to a host with a higher (category) order.
So, in the previous example, instances running upon $h_1$ are not allowed to moved to $h_2$,
and $P$ is not allowed to move to VM $i_2$.
Operator $ i:[\![ I ]\!].M_I$ defines the host machine upon which $I$ runs, 
$M_I$ denotes the cache pages allocated to instance $I$.
To ensure that no cache is shared among different VM instances,
we require that for any host $h$ upon which any $I_1$ and $I_2$ are running: 
$I_1 \ne I_2 \Rightarrow M_{I_1} \cap M_{I_2} = \emptyset$,
and if an instance $I$ terminates, then set $M_I$ to be $\emptyset$ for future allocation to other instances.
Virtual private network provides connectivity for VM hosts. 
It can be viewed as a virtual network consisting of a set of hosts
where VM instances 
can run and communicate with each other.
The location of a host $h: [\![H]\!]$ indicates a network node of G
in which the VM host machine locates.
%

\subsection{Operational semantics}
In order to incorporate as much parallel executions of events within different nodes as possible, 
we transform the network $G$ into a finite parallel compositions of the form:
\begin{eqnarray*}
G & = & h_1:i_{11}:[\![P_{11}]\!].M_{11} ~\Vert ~ \dots ~\Vert ~ h_1:i_{1j_1}:[\![P_{1j_1}]\!].M_{1j_1} \\
&\Vert & h_2:i_{21}:[\![P_{21}]\!].M_{21} ~\Vert ~ \dots ~\Vert ~ h_2:i_{2j_2}:[\![P_{2j_2}]\!].M_{2j_2} \\
&\Vert & \dots ~~ \dots ~~ \dots ~~\\
&\Vert & h_m:i_{m1}:[\![P_{m1}]\!].M_{m1} ~\Vert ~ \dots ~\Vert ~ h_m:i_{mj_m}:[\![P_{mj_m}]\!].M_{mj_m}
\end{eqnarray*}
where $h_k$ denotes the identifier of a host, $i_{kj}$ denotes the VM instance located at $h_k$. Each component $h_k:i_{kj}:[\![P_{kj}]\!].M_{kj}$ is considered as a \emph{decomposition} of $G$.
We argue such decomposition is well-defined by applying the following rules of structural equivalence:
\begin{eqnarray*}
G ~\Vert~ G' &\equiv& G' ~\Vert~ G \\
(G ~\Vert ~ G') ~\Vert~ G'' &\equiv& G ~\Vert ~ (G' ~\Vert~ G'') \\
h: [\![I ~\Vert ~ I']\!] &\equiv&  h: [\![I' ~\Vert ~ I]\!]  \equiv h: [\![I ]\!]~\Vert ~ h:[\![ I']\!] \\
h: [\![(I ~\Vert ~ I') ~\Vert ~ I'']\!] &\equiv&  h:  [\![I ~\Vert ~ (I' ~\Vert ~ I'')]\!] \\
i: [\![P ~\Vert ~ P']\!].M_i &\equiv&  i: [\![P' ~\Vert ~ P]\!].M_i  \equiv i: [\![P ]\!].M_i ~\Vert ~ i:[\![ P']\!].M_i \\
i: [\![(P ~\Vert ~ P') ~\Vert ~ P'']\!].M_i &\equiv&  i:  [\![P ~\Vert ~ (P' ~\Vert ~ P'')]\!].M_i
\end{eqnarray*}
We now define the operational semantics of \seccloud in terms of multiset labelled transition system
$\langle \vec{\Gamma}, \Sigma, \Longrightarrow \rangle$, where:
\begin{itemize}
\item $\vec{\Gamma}$ is a vector of configurations of a VPN $G$ regarding the vector of decomposition of $G$. 
	A configuration $\Gamma$, regarding a single component (say process $P$) of the decompositions of $G$, is defined as a tuple $(\sigma, \delta, I, H)$, where:
	\begin{itemize}
	\item $\sigma: \var_P \mapsto \INT$ denotes the store;
	\item $\delta: \addr_P \mapsto (\INT \cup \{\emptyset\})$ defines the possible world regarding cache; 
	\item $I$ specifies the owner (the VM instance identifier) of process $P$;
	\item $H$ specifies the host (location) of the VM instance of process $P$.
	\end{itemize}
\item $\Sigma$ is a set of operating events which the processes can perform;
\item  $\Longrightarrow \subseteq \vec{\Gamma} \times \vec{\Sigma} \times \vec{\Gamma}$ is the multiset transition relations: $\vec{\Sigma}$ denotes a vector of operating events for the vector of components.
\end{itemize}

The action rules of the operational semantics of \seccloud is presented in Table~\ref{tbl:semantics}.
\begin{table}[!h]
\begin{center}
\scalebox{0.9}{\begin{tabular}{ll}
\hline \\
Store & $S ::= \{\} \mid \{\var \mapsto \INT\} \cup S$
\\
Cache & $M ::= \{\} \mid \{\addr \mapsto (\INT \cup \{\emptyset\}) \} \cup M$
\\ \\
Stop &  
$
\langle \istop, \Gamma \rangle \longrightarrow \Gamma \lbrack \delta(\addr_P) = \emptyset \rbrack
$
\\ \\
Skip &
$
\langle \iskip, \Gamma \rangle \longrightarrow \Gamma
$
\\  \\
Sleep &  
$
\inference{\Delta t>0}
{\langle \sleep(\Delta t), \Gamma \rangle \longrightarrow 
 \langle \sleep(\Delta t-1), \Gamma \rangle}
 ~~~
\inference{\Delta t=0}
 {\langle \sleep(\Delta t), \Gamma \rangle \longrightarrow \Gamma}
$
\\  \\
Assignment &  
$
\inference{\Gamma \vdash \exp \Downarrow v}
{\langle x:=\exp, \Gamma \rangle \longrightarrow 
\Gamma \lbrack \sigma(x)=v \rbrack}
$
\\  \\
Move &  
$
\inference{\Gamma \vdash I \Downarrow i}
{\langle \move_P(i'), \Gamma \rangle \longrightarrow 
\Gamma \lbrack I=i'\rbrack}
$
~~~
$
\inference{\Gamma \vdash H \Downarrow h}
{\langle \move_I(h'), \Gamma \rangle \longrightarrow 
\Gamma \lbrack H=h'\rbrack}
$
\\  \\
Seq & 
$
\inference{\langle P, \Gamma  \rangle \longrightarrow \langle P',\Gamma'  \rangle}
{\langle P;Q, \Gamma \rangle \longrightarrow \langle P';Q,\Gamma'  \rangle}
~~~~~~~~
\inference{\langle P, \Gamma \rangle \longrightarrow \Gamma' }
{\langle P;Q, \Gamma \rangle \longrightarrow \langle Q,\Gamma'  \rangle}
$
\\ \\
Branch &
$
\inference{
\Gamma \vdash b \Downarrow \true
}
{\langle P \lhd b \rhd Q, \Gamma \rangle 
\longrightarrow 
\langle P, \Gamma \rangle
}
~~~~~~
\inference{
\Gamma \vdash b \Downarrow \false
}
{\langle P \lhd b \rhd Q, \Gamma \rangle 
\longrightarrow 
\langle Q, \Gamma \rangle
}
$
\\  \\
Send & 
$\inference{
  \Gamma \vdash \sigma(w) \Downarrow v ~~ w \Leftarrow \addr_a
}
{
  \langle a!(w), \Gamma \rangle
  \longrightarrow
  \Gamma \lbrack \delta(\addr_a)=v \rbrack 
}$
\\ \\
Recv & 
$\inference{
	a \Rightarrow \addr_a ~~
  \Gamma \vdash \delta(\addr_a) \Downarrow v 
}
{
  \langle a?x, \Gamma \rangle
  \longrightarrow
  \Gamma \lbrack \sigma(x)=v \rbrack 
}
$
\\ \\
Loop & 
$\langle b \rhd (P)^*, \Gamma \rangle \longrightarrow 
\langle (P; b \rhd (P)^*) \lhd b \rhd \iskip, \Gamma \rangle$
\\ \\
Parallel & 
$
\inference{P \longrightarrow P'}
{P \Vert Q \longrightarrow P' \Vert Q}
~
\inference{I \longrightarrow I'}
{I \Vert I'' \longrightarrow I' \Vert I''}
~
\inference{H \longrightarrow H'}
{H \Vert H'' \longrightarrow H' \Vert H''}
~
\inference{G \longrightarrow G'}
{G \Vert G'' \longrightarrow G' \Vert G''}
$
\\  \\
\hline
\end{tabular}}
\end{center}
\caption{Operational semantics of \seccloud}
\label{tbl:semantics}
\end{table}
Notations $\Rightarrow$, $\Leftarrow$, $\Downarrow$ denote cache addressing, 
cache allocation, and evaluation respectively.
For instance, $\Gamma \vdash \exp \Downarrow v$ means that
under configuration $\Gamma$, $\exp$ evaluates to value $v$,
$w \Leftarrow \addr_a$ means the cache allocated for expression $w$ locates at $\addr_a$,
and $a \Rightarrow \addr_a$ means cache address of channel $a$ is $\addr_a$,
notation $\addr_P$ is used to denote the cache addresses allocated for $P$.

Store is defined as a mapping from variables to values, 
\ie $\sigma: \var \mapsto \INT$.
Cache is considered as a mapping from addresses (of cache lines) to integers (cached) or $\emptyset$ (flushed), 
\ie $\delta: \addr \mapsto (\INT \cup \{\emptyset\})$.

Action rule of assignment $\langle x:=\exp, \Gamma \rangle$ updates the configuration such that the state of $x$ is the value $v$ of expression $\exp$ after the execution.
Action rule of process moving operator $\move_P (i')$ updates the configuration such that the identifier of instance accommodating $P$ turns to be $i'$ after the execution of the process movement. 
Similar action rule is applied for VM instance movement from one host to another.
Action rule of sending operator $a!(w)$ updates the configuration such that the value stored in cache address ($\addr_a$) is $v$, 
if expression $w$ evaluates to $v$ under configuration before the execution and the address of cache allocated for communicating channel $a$ (to store expression $w$) is $\addr_a$.
Rule of receiving operator $a?x$ updates the configuration such that the state of variable $x$ is $v$,
if the value stored in cache address $\addr_a$ of the communicating channel is $v$ under the configuration before receiving the data.
In the cross-VM communications over today's common virtualised platforms, 
the cache transmission scheme requires the sender and receiver could only communicate by interleaving their executions for security concerns. 
In order to capture timing behaviour of cache-related operations,
we consider the cache-related operations, such as communicating, 
as time-sensitive behaviours whose lasting time is recorded. 
The events are therefore considered as 
either non-delayable (time-insensitive, time does not progress) actions or 
delayable (time-sensitive, involves passing of time $\Delta t$) behaviour.
We use $\Delta t(e)$ to denote the time duration of event $e$ lasting for.

The behaviour of a process component can now be viewed as a set of sequences of \emph{timed runs}.

\begin{definition}[Timed run]
A timed run of a component of $G$ is a sequence of timed configuration event pairs
leading to a final configuration:
\[
\lambda = \langle \Gamma_0, (e_0, \Delta t_0)\rangle \to 
	    \dots \to 
	    \langle \Gamma_n, (e_n, \Delta t_n)\rangle \to
	    \Gamma
\]
where,
$\Gamma_0$ and $\Gamma$ denote the initial and final configuration respectively.
$\forall i \in \{ 0,n \}$, 
$e_i$ is an event will take place with passing time $\Delta t_i$ under configuration $\Gamma_i$;
if $\Delta t_i=0$, $e_i$ is considered as an immediate event,
if $\Delta t_i>0$, $e_i$ is considered as a time-sensitive event with lasting time $\Delta t_i$.
\end{definition}

\section{Attacker Model}
\label{sec:attacker}
We consider computing environment where malicious tenants can use observations 
on process executions and on the usage of shared cache to induce information about victim tenants.
We assume the service provider and the applications running on the victim's VM are trusted.
Let us consider the attacker owns a VM and runs a program on the system, 
and the victim is a co-resident VM that shares the host machine with the attacker VM.
In particular, 
there are two ways in which an attacker may learn secrets from a victim process:
by probing the caches set and measuring the time to access the cache line (through the cache timing channel),
and by observing how its own executions are influenced by the executions of victim processes.

\paragraph{Cache timing side channel (\textbf{C1}).}
In the virtualised computing environment, 
different VMs launch on the same CPU core, 
CPU cache can be shared between the malicious VM and the target ones. 
Consider the malicious VM program repeatedly accesses and monitors the shared cache
to learn information about the sensitive input of the victim VM.
Specifically, timings can be observed from caches and are the most common side channels 
through which the attacker can induce the sensitive information of the victim VM.
Intuitively, such cache side channel can be viewed as a communication channel,
the victim and the attacker can be viewed as the sender and the receiver respectively.
We assume that the attacker knows the map between memory locations and cache sets,
and is able to perform repeated measurements to record when the victim process accesses the target cache line.
We focus on cache timing side channel, 
all other side-channels are outside our scope.
\begin{example}
\label{eg:lang}
Consider the scenario presented in Fig.~\ref{fig:eg}, 
where $\mathtt{VM}_1$ (victim VM, labelled $i_1$) and $\mathtt{VM}_2$ (attacker VM, labelled $i_2$) be two instances running upon $\mathtt{Host}_1$ (labelled $h_1$); 
victim processes $P$ and $Q$, running over $\mathtt{VM}_1$, are communicating to each other:
$P$ generates a key and sends it to $Q$,
$Q$ encrypts a message using the received key and sends the encrypted message to $P$,
$P$ receives the message and decrypts it;
and attacker process $R$, running over $\mathtt{VM}_2$, 
keeps probing cache address of the key.
 \begin{figure}
\centering
\includegraphics[scale=0.4]{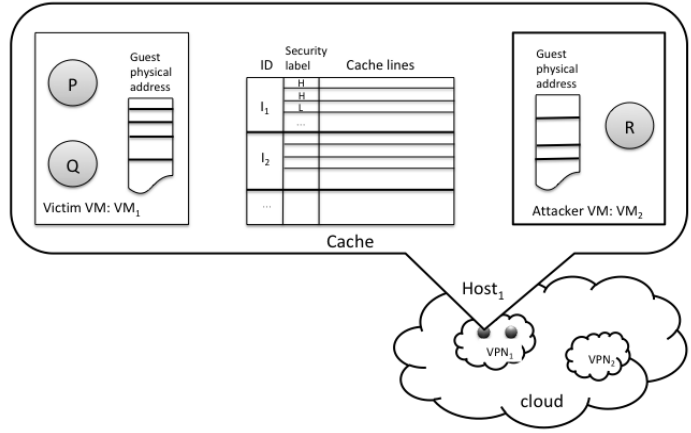}
\caption{Composition structure of the example in networks}
\label{fig:eg}
\end{figure}
Let $\mathtt{keyGen}$, $\mathtt{encrypt}$ and $\mathtt{decrypt}$ denote 
 the function of generating a key, encryption and decryption respectively,
 and assume $\iread(\addr)$ is a function probes cache address $\addr$: returns 1 if it is available and returns $-1$ otherwise.
We present the model in our language as follows.
 \begin{eqnarray*}
  \mathtt{VPC}_1 &\defeq& h_1:[\![\mathtt{Host}_1 ]\!] ~\Vert ~ h_2:[\![\mathtt{Host}_2]\!]\\
  \mathtt{Host}_1 &\defeq& i_1:[\![\mathtt{VM}_1]\!].M_{\mathtt{VM}_1} ~\Vert~ i_2:[\![\mathtt{VM}_2]\!].M_{\mathtt{VM}_2}\\
  \mathtt{VM}_1 &\defeq& [\![P \Vert Q]\!]  ~~~~~~~~
  \mathtt{VM}_2 \defeq  [\![R]\!]\\
  P &\defeq&  x:=\mathtt{keyGen}() ;~~ 
  										\mathtt{key}!x \to P'\\
	Q &\defeq& \mathtt{key}? y \to Q' \\
  Q' &\defeq&  m_1:=\mathtt{encrypt}(y, \text{``message''}) ; ~~
  						\mathtt{msg}!m_1 \to \istop		\\
   P' &\defeq& \mathtt{msg}?m_2 \to P'' \\
  P'' &\defeq& m := \mathtt{decrypt}(x, m_2) \to \istop\\
  R &\defeq& \true \rhd (z:=\iread(\addr_{\mathtt{key}}))^*
 \end{eqnarray*}
The communication time between $P$ and $Q$ is affected by the value of the key generated.
There is information flow from victim VM to malicious VM through cache side channel.
\end{example}

\paragraph{Leakage through observations on process executions (\textbf{C2}).}
Next let us focus on processes running upon a particular VM instance. 
Consider the VM instance, upon which processes are running, as a communication channel having inputs and outputs, 
which is relative to information flow from victim user to a malicious one.
The victim user controls a set of higher level inputs with confidential data.
The attacker is an observer who may have control of lower level inputs.
He has partial observation on executions of the process, 
but does not have any access to the confidential data. 
Specifically, the weak attacker can observe the public result, 
\ie the final public output of the programs, 
while the strong attacker can observe the low state after each execution step of the processes. 
We consider a process has access to confidential data via higher level inputs. 
The attacker tries to collect and deduce some of the secure information 
about higher level inputs by varying his inputs and observing the execution of the process.
\begin{example}
\label{eg:lang2}
Consider process $P$ and $R$ are running upon a VM instance $\mathtt{VM}$. Process $P$ inputs a password through channel $\mathtt{pwd}$ into $H$-level variable $x$, and updates $L$-level variable $y$ to be $1$ if $x$ is odd and to be $0$ if $x$ is even. Process $R$ output variable $y$ through channel $\mathtt{res}$:
 \begin{eqnarray*}
 \mathtt{VM} &\defeq& [\![P \Vert R]\!]\\
 P &\defeq& \mathtt{pwd}? x;~~ y:=1 \lhd (x\mod 2 == 1) \rhd y:=0 \to \istop \\
 R &\defeq& \mathtt{res}! (y)  \to \istop
 \end{eqnarray*}
Assume $L \sqsubset H \in \L$. Clearly there are implicit flows from $x$ to $y$ by observing $L$-level output of the process.
\end{example}
\section{Information Flow Policy}
\label{sec:policy}

Information flow is controlled by means of security labels and flow policy integrated in the language. 
Each of the identifiers, information container, is associated with a security label.
Identifiers can refer to variables, communication channels, 
and can refer to entities such as files, devices in a concrete level.
The set of the security labels forms a security lattice regarding their ordering.
We study the system flow policy which prevents information flow leakage
from high-level objects to lower levels 
and from a target instance to a malicious one 
via observing process executions and 
cache usage (by measuring the time of accessing cache lines during communications).

In general, information flow policies are proposed to ensure that secret information does not influence publicly observable information. 
An ideal flow policy called Non-interference (NI)~\cite{GoguenMeseguer82} 
is a guarantee that no information about the sensitive inputs can be obtained 
by observing a program's public outputs, 
for any choice of its public inputs.
Intuitively, the NI policy requires that
low security users should not be aware of the activity of high security users 
and thus not be able to deduce any information about the behaviours of the high users.
On the one hand,
for processes running upon a particular VM instance (regarding \textbf{C2}), 
the NI policy can be applied 
to control information flow from high-level input to low-level output,
where state of sensitive information container (s.a. high-level variables) and 
observations on behaviours of public information container (s.a. low-level variables)
are viewed as the high input and low output respectively. 
On the other hand, 
for processes running upon different VM instances (regarding \textbf{C1}), 
we adapt the NI policy here in order to 
control the information flow from processes running upon victim instance to malicious one
through cache side channel. 
Consider the cache side channel as a communication channel,
the cache lines accessed by the victim instance and by the malicious instance 
are viewed as high level input and low level input respectively,
and the observations on the victim cache usage (s.a. hits/misses) are considered as low level outputs.
\emph{Cache flow non-interference} demands the changing of the cache lines accessed by the victim process (high inputs) does not affect the public observations on the cache usage (low outputs).
Informally, cache flow interference happens if the usage (we focus on the accessing time) of the cache lines accessed by one victim process affects the usage of the cache lines accessed by attacker VM processes. 
Fig.~\ref{fig:ni4cache} presents some intuition of the cache NI policy discussed above.
 \begin{figure}
 \centering
 \includegraphics[scale=0.55]{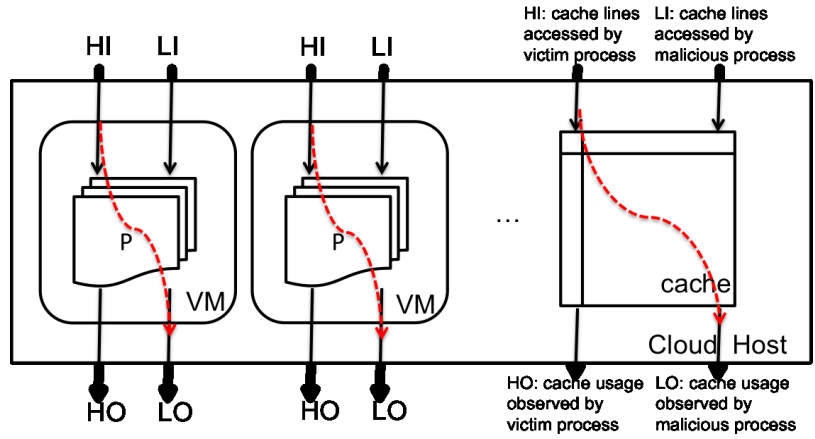}
 \caption{Cache non-interference}
 \label{fig:ni4cache}
 \end{figure}

Formally, the policy of flow non-interference can be considered in terms of 
the equivalence relations on the system behaviours from the observer's view,
including the state evolution of information container and 
the timing behaviour of cache accessing.
This is due to the fact that the system behaviours are modelled as timed runs 
with security classification of identifiers and 
with timing considerations when accessing caches 
during the process communications in our model.
%
\begin{definition}[Flow security environment]
\label{def:sec-env}
Let $\mathcal L$ be a finite flow lattice, 
$\sqsubseteq$ denote the ordering relation of $\mathcal L$,
$I$ denote the set of VM instances running upon any host $h$,
$\Omega_H$ and $\Omega_I$ denote a set of categories for hosts
and sub-categories for instances respectively.
The flow security environment is considered as:
\[
\Xi: (\tau_v, \tau_c, \alpha_c, \beta_h, \beta_i),
\]
where $\tau_v: \var \mapsto_{\tau_v} \L$, 
$\tau_c: \CLines \mapsto_{\tau_c} \L$, 
$\alpha_c: \CLines \mapsto_{\alpha_c} I$,
$\beta_h: \hosts \mapsto_{\beta_h} \Omega_H$,
$\beta_i: \instances \mapsto_{\beta_i} \Omega_I$.
Furthermore, we say $\Xi \sqsubseteq \Xi'$ \emph{iff}
$\forall x \in \var$, $l \in \CLines$, $i \in \instances$ and $h \in \hosts$:
\begin{eqnarray*}
&& \Xi(\tau_v(x)) \sqsubseteq \Xi'(\tau_v(x)) ~\land~ 
\Xi(\tau_c(l)) \sqsubseteq \Xi'(\tau_c(l)) \\
& \land &
\Xi(\beta_h(h)) \sqsubseteq \Xi'(\beta_h(h)) ~\land ~
\Xi(\beta_i(i)) \sqsubseteq \Xi'(\beta_i(i)),
\end{eqnarray*}
and for $t \in \L$, $\omega_I \subseteq \Omega_I$,
and $\omega_H \subseteq \Omega_H$,
we say $\Xi \sqsubseteq (t, \omega_I, \omega_H)$ \emph{iff}:
\[
\Xi(\tau_v(x)) \sqsubseteq t ~\land ~ \Xi(\tau_c(l)) \sqsubseteq t
~\land ~ \Xi(\beta_h(h)) \sqsubseteq \omega_H
~\land ~ \Xi(\beta_i(i)) \sqsubseteq \omega_I,
\]
where we abuse notation $\Xi(\tau_v(x))$, $\Xi(\tau_c(l))$, 
$\Xi(\beta_h(h))$ and $\Xi(\beta_i(i))$
to denote $\tau_v(x)$, $\tau_c(l)$, $\beta_h(h)$ and $\beta_i(i)$ 
in environment $\Xi$ respectively.
\end{definition}
\begin{definition}[$(t, \omega_I, \omega_H)$-equivalent configuration $=_{\Xi,(t, \omega_I, \omega_H)}$]
Consider processes running upon hosts of VPN $G$,
let $\Xi$ be a security environment,
$t \in \L$ be a security level,
$\omega_H \subseteq \Omega_H$ be a category and 
$\omega_I \subseteq \Omega_I$ be a sub-category.
For any $x \in \var$, $l \in \CLines$, 
assume $i$ and $h$ are the instance and host which $x, l$ belong to,
we define store $(t, \omega_I, \omega_H)$-equivalence under $\Xi$ as follows:
$\sigma_1(x)=_{\Xi,(t, \omega_I, \omega_H)}\sigma_2(x)$
  \emph{iff}:
\[  ( \Xi(\tau_v(x)) \sqsubseteq t ~\land~ 
\Xi(\beta_i(i))  \sqsubseteq \omega_I ~\land~ 
\Xi(\beta_h(h))  \sqsubseteq \omega_H)
  \Rightarrow
  \sigma_1(x) = \sigma_2(x),
\]
and cache line $(t, \omega_I, \omega_H)$-equivalence under $\Xi$ as follows:
$\delta_1(l)=_{\Xi,(t, \omega_I, \omega_H)}\delta_2(l)$
  \emph{iff}:
\[  (\Xi(\tau_c(l)) \sqsubseteq t  ~\land~ 
\Xi(\beta_i(i))  \sqsubseteq \omega_I ~\land~ 
\Xi(\beta_h(h))  \sqsubseteq \omega_H)
  \Rightarrow
  \delta_1(l) = \delta_2(l).
\]
Furthermore, given two configurations $\Gamma_1 = (\sigma_1, \delta_1, i_1, h_1)$ and 
$\Gamma_2 = (\sigma_2,\delta_2, i_2, h_2)$,
we say $\Gamma_1 =_{\Xi,(t, \omega_I, \omega_H)} \Gamma_2$ \emph{iff}:
\begin{eqnarray*}
&& (\forall x \in \var. \sigma_1 (x) =_{\Xi,(t, \omega_I, \omega_H)} \sigma_2 (x)) 
\land~
(\forall l \in \CLines. \delta_1 (l) =_{\Xi,(t, \omega_I, \omega_H)} \delta_2 (l)) \\
&\land & (\beta_i(i_1) \sqsubseteq \beta_i(i_2)) 
~\land~
(\beta_h(h_1) \sqsubseteq \beta_h(h_2)) .
\end{eqnarray*}
\end{definition}
\begin{definition}[Strong bisimulation 
${\buildrel\rm s\over\sim_{\Xi,(t, \omega_I, \omega_H)}}$ and 
weak bisimulation ${\buildrel\rm w\over\sim_{\Xi,(t, \omega_I, \omega_H)}}$]
\label{def:simulation}
Consider two timed runs running upon host $h$
under security environment $\Xi$: 
\[
\lambda = \langle \Gamma_0, (e_0, \Delta t_0) \rangle \to \dots \to \langle \Gamma_n, (e_n, \Delta t_n) \rangle \to \Gamma
\]
\[
\lambda' = \langle \Gamma'_0, (e'_0, \Delta t'_0) \rangle \to \dots \to \langle \Gamma'_n, (e'_n, \Delta t'_n) \rangle \to \Gamma'
\]
$\forall \Gamma_0, \Gamma'_0$ such that $\Gamma_0 =_{\Xi,(t, \omega_I, \omega_H)} \Gamma'_0$,
we say $\lambda$ and $\lambda'$ are \emph{strong $(t, \omega_I, \omega_H)$-bisimilar} to each other, 
\ie  
$\lambda\ {\buildrel\rm s\over\sim_{\Xi,(t, \omega_I, \omega_H)}}\ \lambda'$, 
iff:
\[
\forall j \in \{0...n\}.(\Gamma_{j} =_{\Xi,(t, \omega_I, \omega_H)} \Gamma'_{j}) \land (\Delta t_j = \Delta t'_j);
\]
and say $\lambda$ and $\lambda'$ are \emph{weak $(t, \omega_I, \omega_H)$-bisimilar} to each other, \ie $\lambda\ {\buildrel\rm w\over\sim_{\Xi,(t, \omega_I, \omega_H)}}\ \lambda'$, iff:
\[
(\Gamma =_{\Xi,(t, \omega_I, \omega_H)} \Gamma') \land (\sum^n_{j=0} \Delta t_j = \sum^n_{j=0} \Delta t'_j).
\]

\end{definition}
\begin{definition}[Cache flow security policy]
\label{def:policy}
Given a security level $L \in \L$, $\omega_H \subseteq \Omega_H$, and $\omega_I \subseteq \Omega_I$
a VPN $G$ under security environment $\Xi$ is considered \emph{strong cache flow secure} 
\emph{iff}: 
\[
\forall \lambda, \lambda' \in \Lambda. 
(\Gamma_0 =_{\Xi, (L, \omega_I, \omega_H)} \Gamma'_0 
\Rightarrow 
\lambda\ {\buildrel\rm s\over\sim_{\Xi,(L, \omega_I, \omega_H)}}\ \lambda'),
\]
where $\Gamma_0$ and $\Gamma'_0$ denote the initial configuration of $\lambda$ and $\lambda'$ respectively,
$\Lambda$ denotes all runs of components of $G$.
Similarly, the definition of \emph{weak cache flow secure} can be given.
\end{definition}

\begin{example}
\label{eg:policy}
Consider the model presented in Example~\ref{eg:lang}. 
Let $\tau_v(x) = \tau_v(y) = H$,
$\tau_v(m_1) = \tau_v(m_2) = M$,
$\tau_v(z) = L$,
and $L \sqsubset M \sqsubset H \in \L$.
Let us assume $\beta_i(\mathtt{VM}_1) = \beta_i(\mathtt{VM}_2) = \omega_I$,
and $\beta_h(\mathtt{Host}_1) = \omega_H$.
Communicating cache channels are thus assigned with security labels regarding the data they transmit:
$\tau_c(\addr_{\mathtt{key}}) = H$,
$\tau_c(\addr_{\mathtt{msg}}) = M$.
Note that the state of variable $z$ depends on the state of cache address of $\mathtt{key}$, and is affected by the communication time for data transmission between $P$ and $Q$.
Therefore the model does not satisfy the cache flow security policy since for any given two runs, both the timing condition and configuration equivalent condition of ${\buildrel\rm w\over\sim_{\Xi,(L, \omega_I, \omega_H)}}$ are not guaranteed to be satisfied.
\end{example}
In order to close the cache timing channel,
we consider the communication as a scenario of sending and receiving processes running in parallel 
with certain time interleaving data transmission scheme:
\[
\inference{
  \Gamma \vdash w \Downarrow v ~~ 
  w \Leftarrow \addr_a ~~ 
  a \Rightarrow \addr_a ~~
  t < T 
}
{
  \begin{array}{l}
   \langle \sleep(T-t) \to P ~\lhd~ (t:=\Delta t(a!w ~\Vert~ a?x) < T) ~\rhd~ \istop, 
	  \Gamma \rangle\\
   ~ \buildrel a(w) \over \longrightarrow ~
   \langle P, \Gamma \lbrack \sigma(x)=v,\delta(\addr_a)=\emptyset \rbrack \rangle 
  \end{array}
}
~~~ (1)
\]
The communicating procedure needs to complete in $T$ time units (together with sleeping time)
and then behaves as $P$; 
the value of $w$ is sent to variable $x$ via channel $a$,
channel $a$ and the relevant cache lines are then tagged as $\tau_c(w)$.
The communication will be considered as failed if $T$ time units have passed but the communication has not completed yet. 
Fixed completion time $T$ prevents the timing leakage introduced by the cache channel communication.
\begin{example}
\label{eg:timing}
Consider the model presented in Example~\ref{eg:lang}, we rewrite the communicating procedure as follows: 
 \begin{eqnarray*}
 && x:=\mathtt{keyGen}() ;\\
 &&  \sleep(5-t) \to P' \lhd (t:=\Delta t(\mathtt{key}!x ~\Vert ~ \mathtt{key}? y) < 5) \rhd \istop\\
&& m_1:=\mathtt{encrypt}(y, \text{``message''}) ; \\
 && \sleep(5-t) \to Q' \lhd (t:=\Delta t(\mathtt{msg}!m_1 ~ \Vert ~ \mathtt{msg}?m_2)<5) \rhd \istop
 \end{eqnarray*}
The timing condition of weak $(t, \omega_I, \omega_H)$-bisimulation specified in Definition~\ref{def:simulation} is now ensured, while the configuration condition is still violated. This will be addressed in next Section.
\end{example}
\section{Flow Security Type System for \seccloud}
\label{sec:type}

In order to make the low observation and cache accessing time of the executions be high input independent,
the variables and cache lines are associated with security labels, 
VMs and hosts are assigned to categories,
rules (semantic + typing) are required to ensure that: 
no information flows to lower level objects,
no cache is shared among different VM instances,
processes (c.f. instances) are not allowed to move from a lower order instance (c.f. host) to a higher one,
and cache related operations in communications between processes are forced to be completed in certain time,
and hence the cache flow policy is enforced.

For a process $P$ (\wrt a component of $G$), we consider the type judgements have the form of:
\[ (\tau, \omega_I, \omega_H) \vdash \Xi \{P\} \Xi'\] 
where the type $(\tau, \omega_I, \omega_H)$ denotes the (environment) \emph{counter security levels} of the communication channel/variables and \emph{and counter categories} of VMs/hosts
participated in the branch events being executed for the purpose of 
eliminating implicit flows from the \emph{guard}.
$\Xi$ and $\Xi'$ describe the type environment which
hold before and after the execution of $P$.
In general, notation: 
\[
\Xi \vdash (\exp:_{\tau_v} t_e, 
 l:_{\tau_c} t_l,
 l:_{\alpha_c} i_l,
 i:_{\beta_i} \omega_I,
 h:_{\beta_h} \omega_H)
\]
describes that under type environment $\Xi$,
expression $\exp$ and (the address of) cache line $l$ 
has type $t_e$ and $t_l$ respectively,
$l$ is allocated to VM instance $i_l$, 
instance $i$ is assigned to a category $\omega_I$,
and host $h$ is assigned to a category $\omega_H$.
The type of an expression including boolean expression is defined by 
taking the least upper bound of the types of its free variables as standard:
\[
\Xi \vdash \exp:_{\tau_v}t ~\text{iff}~ 
t=\sqcup_{x \in \mathtt{fv}(\exp)} \Xi(\tau_v(x)).
\]
All memory, caches and channels written by a t-level expression becomes tagged as t-level.
Let $\var_P$ and $\CLines_P$ denote a set of variables defined in and cache lines allocated to process $P$.
Typing rules for processes with security configuration are presented in Table~\ref{tbl:rules}.
\begin{table}[h!]
\begin{center}
\scalebox{0.85}{\begin{tabular}{ll}
\hline\\
 (\tsub) & ~~~
 $\inference{ \tau_1 \vdash \Xi_1 \{P\} \Xi'_1}
 {\tau_2 \vdash \Xi_2 \{P\} \Xi'_2 } $~~~~
 $\tau_2 \sqsubseteq \tau_1,~ \Xi_2 \sqsubseteq \Xi_1,~ \Xi'_1 \sqsubseteq \Xi'_2$
 \\\\
(\tass) & ~~~
$\inference{\Xi \vdash \exp:_{\tau_v} t}
{ (\tau, \omega_I, \omega_H) \vdash \Xi \{x:=\exp\} \Xi'(x \mapsto_{\tau_v} \tau \sqcup t)}$
\\\\
(\tstop) & ~~~
$(\bot_{\mathcal L},\bot_{\Omega_I}, \bot_{\Omega_H}) \vdash \Xi \{ \istop_P \} 
\Xi'(\{l \mapsto_{\tau_c} \bot \mid \forall l \in \CLines_P\})$
\\\\
(\tskip) & ~~~
$(\bot_{\mathcal L},\bot_{\Omega_I}, \bot_{\Omega_H}) \vdash \Xi \{ \iskip \}\Xi$
\\\\
(\tmove) & ~~~
$\inference{\Xi \vdash (\{x:_{\tau_v}t_x, l:_{\tau_c}t_l, l:_{\alpha_c}i \mid \forall x\in \var_P, l\in \CLines_P\}),~ i:_{\beta_i} \omega,~ i':_{\beta_i} \omega'}
{\begin{array}{l}(\tau, \omega_I, \omega_H) \vdash \Xi \{ \move_P (i') \} \Xi' 
 (\{x\mapsto_{\tau_v}t_x\sqcup\tau, l\mapsto_{\tau_c}t_l\sqcup\tau, l\mapsto_{\alpha_c}i' \cr
 ~~~~~~~~~ ~~~~~~~~~ \mid \forall x\in \var_P, l\in \CLines_P\},~ 
 i \mapsto_{\beta_i} \sqcup\{\omega, \omega', \omega_I\}) 
 \end{array}}$
\\\\
& ~~~
$\inference{
\begin{array}{l}
h: [\![I]\!].M_I ~~~~~ I= \langle i:[\![P_1]\!],~  i:[\![P_2]\!],~ \dots,~  i:[\![P_n]\!] \rangle \cr
h': [\![I']\!].M_{I'} ~~ I'= \langle i':[\![P'_1]\!],~  i':[\![P'_2]\!],~ \dots,~  i':[\![P'_n]\!] \rangle \cr
\Xi \vdash h:_{\beta_h} \omega ~~~~
\Xi \vdash h':_{\beta_h} \omega' \cr
(\tau, \omega_I, \omega_H) \vdash \Xi_1\{\move_{P_1} (i') \}\Xi'_1 ~~~ 
\dots ~~~
(\tau, \omega_I, \omega_H) \vdash \Xi_n\{\move_{P_n} (i') \}\Xi'_n
\end{array}
}
{(\tau, \omega_I, \omega_H) \vdash \langle \Xi_1, \Xi_2, \dots, \Xi_n \rangle \{ \move_I (h') \} 
 \langle \Xi'_1,  \Xi'_2, \dots, \Xi'_n\rangle (h \mapsto \sqcup \{\omega, \omega', \omega_H\})}$
\\\\
(\tseq) & ~~~
$\inference{(\tau, \omega_I, \omega_H) \vdash \Xi \{P\} \Xi' 
~~~~
(\tau, \omega_I, \omega_H) \vdash \Xi' \{Q\} \Xi''}
{ (\tau, \omega_I, \omega_H) \vdash \Xi \{P; Q\} \Xi''} $
\\\\
(\tsend) & ~~~
$\inference{\Xi \vdash w:_{\tau_v} t ~~ w \Leftarrow \addr_a}
{ (\tau, \omega_I, \omega_H) \vdash \Xi \{a!w\} 
  \Xi'(\addr_a \mapsto_{\tau_c} \tau \sqcup t)}$
\\\\
(\trecv) & ~~~
$\inference{\Xi \vdash \addr_a:_{\tau_c} t}
{ (\tau, \omega_I, \omega_H) \vdash \Xi \{a?x\} \Xi'(x \mapsto_{\tau_v} \tau \sqcup t)}$
 \\\\
(\tbranch) & ~~~
$\inference{\begin{array}{l}         
	\Xi \vdash b:_{\tau_v} t \cr 
	(t \sqcup \tau, \omega_I, \omega_H) \vdash \Xi \{P\} \Xi'_P \cr 
	(t \sqcup \tau, \omega_I, \omega_H) \vdash \Xi \{Q\} \Xi'_Q
            \end{array}}
{ (\tau, \omega_I, \omega_H) \vdash \Xi \{P \lhd b \rhd Q\} \Xi'} $ 
~~ $\Xi'=\Xi'_P \sqcup \Xi'_Q$
\\\\
(\tloop) & ~~~
$\inference{\begin{array}{l}
\Xi_i \vdash b:_{\tau_v} t_i \cr
(t_i \sqcup \tau,\omega_I, \omega_H) \vdash \Xi_i \{P\} \Xi'_i ~~ i=0, \dots, n
\end{array}}
{ (\tau, \omega_I, \omega_H) \vdash \Xi \{b \lhd (P)^*\} \Xi'_n} $ 
~~ $\Xi_0=\Xi, ~ \Xi_{i+1} = \Xi'_i \sqcup \Xi, ~ \Xi_{n+1} = \Xi'_n$
\\\\
(\tpar) & ~~~
$\inference{ 
  \begin{array}{l}
    (\tau, \omega_I, \omega_H) \vdash \Xi \{P\} \Xi' ~~~ 
    (\tau, \omega_I, \omega_H) \vdash \Xi \{Q\} \Xi''  \cr
    \Xi' \vdash (\{x :_{\tau_v} t'_x, l:_{\tau_c} t'_l \mid \forall x\in \var, l\in\CLines\}) \cr
    \Xi'' \vdash (\{x :_{\tau_v} t''_x, l:_{\tau_c} t''_l i \mid \forall x\in \var, l\in\CLines\}) 
  \end{array}
}
{(\tau, \omega_I, \omega_H) \vdash \Xi \{P \Vert Q\} 
\Xi'''(\{x\mapsto_{\tau_v} t'_x \sqcup t''_x, 
	l\mapsto_{\tau_c} t'_l\sqcup t''_l 
      \mid \forall x\in \var, l \in \CLines\})
}$ 
\\\\
\hline \\
\end{tabular}}
\end{center}
\caption{Typing rules for processes with security configuration.}
\label{tbl:rules}
\end{table}
Rule \tstop for stopping a process $P$ ensures the cache lines allocated to $P$ to be flushed as empty with a label of system low $\bot$ for future use.
Rule for assignment \tass ensures the type of variable $x$ to be the least upper bound of 
the type of the assigned expression $\exp$ and the counter level $\tau$.
Rule \tmove regarding moving a process from one instance to another  
enforces the type of each variable and cache line involved in process $P$
to be the least upper bound of the type of itself and the counter level $\tau$,
and ensures its owner to be the VM identifier $i'$ where the process is moving to,
and ensures category of the moving-to-instance to be the least upper bound of that of itself and the moving-from-instance;
rule \tmove regarding moving an instance from one host $h$ to another $h'$ essentially moves all processes running upon the instance to $h'$ such that each of which follows the rule of process moving.
Rule for composition \tseq ensures the typing environment in terms of the sequential event in a compositional way.
Rule for sending operator \tsend ensures that the type of the cache channel for sending expression $w$
is to be the least upper bound of the type of $w$ and the counter level,
and rule for receiving operator \trecv ensures that the type of the receiving variable $x$
is to be the least upper bound of the type of channel and the counter level.
Rule for branch event \tbranch specifies that the typing environment is required to be enforced
in terms of the branch body with the counter level being 
the least upper bound of the current counter level and the type of the guard.
Rule for loop operator \tloop calculates the least fixed point of the type environment transition function on the security lattice.
Rule for parallel operator \tpar ensures that the type of each information container
(including process variables and cache lines)
is to be the least upper bound after the execution of each parallel process.

In overall, the derivation rule $(\tau, \omega_I, \omega_H) \vdash \Xi \{P\} \Xi'$ ensures that:
\begin{itemize}
\item variables and cache lines whose final types in $\Xi'$ are less than $\tau$ must not be changed by $P$;
\item the final value of a variable (or a cache line) say $x$ whose final type is $\Xi'(x)=t$ must not depend on the initial values of those variables (or cache lines) say $z$ whose initial type $\Xi(z)$ is greater than $t$.
\item processes (c.f. instances) belonging to a higher order category instances (c.f. host)
must not move to an instance (c.f. a host) with a lower order category.
\end{itemize}
\begin{theorem}[Monotonicity of the type environment transition function]
\label{theo:mono}
Given $\L$ and $P$, for all $\tau$, $\omega_I$, $\omega_H$ and $\Xi$, 
the type environment transition function: 
$\T_{P,\L} (\Xi,(\tau,\omega_I,\omega_H)) \mapsto \Xi'$ 
regarding $(\tau,\omega_I,\omega_H) \vdash \Xi \{P\} \Xi'$
is monotone.
\end{theorem}
\begin{proof}
The proof is obtained by induction on the semantic structure of the process.
Particularly, for the case of loop, 
the sequences of $\Xi'_0$, $\Xi'_1$, $\Xi'_2$, $\dots$
and $\Xi''_0$, $\Xi''_1$, $\Xi''_2$, $\dots$
form ascending chains with a least upper bound due to the finiteness of $\L$,
so rule \tloop calculates the least fixed point of $\T_{P,\L}$ on $\L$.
\hfill $\Box$
\end{proof}
\begin{definition}[Semantic flow security condition]
\label{def:semantic-sec-cond}
We say the semantic relation of $P$ satisfies $(\tau,\omega_I,\omega_H)$-flow security property (denoted by $\phi_{\tau,\omega_I,\omega_H}$),
written as: 
$(\Xi \{P\} \Xi') \models \phi_{\tau,\omega_I,\omega_H}$,  iff:
\begin{itemize}
\item [i)] for all $\Gamma$, $\Gamma'$, $x$ and $l$:
\[
\langle P, \Gamma \rangle \Downarrow \Gamma' 
~\land~ \Xi' \sqsubset (\tau,\omega_I,\omega_H)
~\Rightarrow~ 
\Gamma(\sigma(x)) = \Gamma'(\sigma(x))  ~\land~ 
\Gamma(\delta(l)) = \Gamma'(\delta(l));
\]
where we abuse notation $\Gamma(\sigma(x))$, $\Gamma(\delta(l))$ 
to denote $\sigma(x)$, $\delta(l)$ in configuration $\Gamma$.
\item [ii)] and for all $t \in \L$, $o \in \Omega_I$, $o' \in \Omega_H$, $\Gamma_1$ and $\Gamma_2$:
\[
\Gamma_1 =_{\Xi, (t, o, o')} \Gamma_2 
~\Rightarrow~ 
\Gamma'_1 =_{\Xi, (t, o, o')} \Gamma'_2.
\]
\end{itemize}
\end{definition}
We say the flow security type system is sound 
if the well-typed system is flow secure, more precisely,
whenever $(\tau,\omega_I,\omega_H) \vdash \Xi \{P\} \Xi'$ then
the semantic relation of $P$ is flow secure.
\begin{theorem}[Soundness of the flow security type system]
\label{theo:soundness}
The type system proposed in Table~\ref{tbl:rules} is sound, \ie
\[
(\tau,\omega_I,\omega_H) \vdash \Xi \{P\} \Xi' \Rightarrow (\Xi \{P\} \Xi') \models \phi_{\tau,\omega_I,\omega_H}.
\]
\end{theorem}
\begin{proof}
Since $\tau \vdash \Xi \{P\} \Xi'$ ensures that:
for all variables and cache lines $x$,
whenever $\Xi'(x) \sqsubset \tau$
then there are no assignments to $x$ in $P$,
and thus i) is clearly satisfied; 
and for any variables $y$, 
such that $\Xi(y) \sqsupset \Xi'(x)$,
the final value of $x$ does not depend on 
the initial values of $y$,
we can then prove for all $t\in \L$, $o \in \Omega_I$, $o' \in \Omega_H$, 
configurations $\Gamma$ and $\Gamma'$ such that 
$\langle P, \Gamma \rangle \Downarrow \Gamma'$, 
whenever $\Gamma_1 =_{\Xi, (\tau,\omega_I,\omega_H)} \Gamma_2 $
then $\Gamma'_1 =_{\Xi, (\tau,\omega_I,\omega_H)} \Gamma'_2$, 
by induction on the structure of the derivation tree.
 \hfill $\Box$
\end{proof}

\begin{theorem}[Flow secure of communications]
\label{theo:capacity}
Given a \seccloud model $\mathbf A$, 
for all $P$ running upon any VM $i$ of any host $h$ in $\mathbf A$,
if $(\Xi \{P\} \Xi') \models \phi_{\tau,\omega_I,\omega_H}$,
then $\mathbf A$ is weak cache flow secure with $L=\tau$.
\end{theorem}
\begin{proof}
The proof includes two parts in terms of the definition of weak cache flow security policy given in Definition~\ref{def:policy}:
channel of timing observable behaviours of victim process, 
and channel of observable low outputs influenced by high inputs.
The first part is ensured by the semantics of communications specified in (1),
which must be completed in certain time,
and the second part is ensured by the soundness of the type system
presented in Theorem~\ref{theo:soundness}.
\hfill $\Box$
\end{proof}

\begin{example}
\label{eg:type}
Consider again the model presented in Example~\ref{eg:lang} and ~\ref{eg:policy}.
It is clear that $(\tau,\omega_I,\omega_H) \vdash \Xi \{R\} \Xi'$ does not hold since the assignment to $z$ make $\Xi'(z) > \Xi(z) $;
while $\Xi \{R\} \Xi' \models \phi_{\tau,\omega_I,\omega_H}$ holds if the communication between $P$ and $Q$ fails and $\Xi \{R\} \Xi' \not \models \phi_{\tau,\omega_I,\omega_H}$ holds otherwise.
\end{example}
\section{Related Work}
\label{sec:related}

This paper relates to the topic of information flow analysis
in virtualised computing systems from perspective of 
formal languages with a concern of cache timing attacks.

Cross-VM side-channel attacks in virtualised infrastructure allowed the attacker to extract information from a target VM and stole confidential information from the victims~\cite{RistenpartTSS09,WuXW12,ZhangJRR12}.
Over the last decade, there have been sustained efforts in exploring solutions
to defend cache channel attacks in virtualised computing environment via the approaches of 
\emph{cache partition} at either \emph{hardware-level}~\cite{WangL07,WangL08,DomnitserJLAP12,KongASZ13}
or \emph{system-level}~\cite{RajNSE09,ShiSCZ11,KimPM12,ZhouRZ16,LiuGYMRHL16},
and the approaches of 
\emph{cache randomisation} via introducing randomization in cache uses through
either new \emph{hardware design}~\cite{WangL06,WangL07,WangL08,LiuL14,LiuHMHTS15,RenFKDD19}
or \emph{compiler-assistant design}~\cite{GodfreyZ14,LiuHMHTS15,CraneHBLF15,RaneLT15}.
Wang and Lee~\cite{WangL07} identified cache interference as the main cause of 
different types of cache side channel attacks.
Security-aware cache designs were proposed to mitigate the cache interference issue:
Partition-Locked cache (PLcache) eliminated cache interference via preventing cache sharing in a flexible way,
and Random Permutation cache (RPcache) allowed cache sharing but randomised cache interference so that no information can be deduced.
More recently, Liu et. al.~\cite{LiuGYMRHL16} developed CATalyst 
to defend cache-based side channel attacks for the cloud computing system. 
Specifically, CATalyst used Cache Allocation Technology (CAT) on Intel processors 
to partition the last-level cache 
into secure and non-secure partitions. 
The secure partition was loaded with cache-pinned secure pages and was software-managed, 
while the non-secure partition can be freely used by any applications and kept hardware-managed. 
Security-sensitive code and data can be loaded and locked by users via mapping them to the cache-pinned pages.
Thus a hybrid hardware-software managed cache was constructed to protect the sensitive code and data. 
However, these efforts mostly require significant changes to the hardware,
hypervisors, or operating systems, 
which make them impractical to be deployed in current cloud data centres.
Formal treatment on flow security policies upon side-channel attack detection, 
leveraging program analysis techniques and relevant tools are still needed
in order to improve the accuracy and applicability of leakage analysis and control
in the virtualised infrastructure.
In this paper, we address the flow security issue in virtualised computing environment from perspective of \emph{programming language} and \emph{program analysis} techniques.

On the other hand, 
there have been lasting investigations on flow property specification and enforcement via approaches of \emph{formal language and analysis}. 
The conception of information flow specifies the
security requirements of the system where no sensitive information should be released to the observer during its executions.
Denning and Denning~\cite{Denning77} first used program analysis to investigate
if the information flow properties of a program
satisfy a specified multi-level security policy.
Goguen and Meseguer~\cite{GoguenMeseguer82} formalised the notion of absence of information flow with the concept of non-interference.
Ryan and Schneider~\cite{RyanS99} took a step towards the generalisation of a CSP formulation of non-interference to handle information flows through the concept of process equivalence.
Security type systems~\cite{VolpanoS97,Myers99,HondaVY00,Pottier02,HuntSands06,CapecchiCDR10,HuntSands11} had been substantially used to formulate
the analysis of secure information flow in programs.
In addition to type-based treatments of secure information
flow analysis for programs, Clark \emph{et. al}
presented a flow logic approach in~\cite{ClarkHH02},
Amtoft and Banerjee proposed a Hoare-like logic for program dependence in~\cite{AmtoftB04}.
Hammer and Snelting~\cite{HammerS09} presented an approach
for information flow control in program analysis
based on program dependence graphs (PDG).
Based on~\cite{HammerS09},
~\cite{TaghdiriSS10} extended the PDG-based flow analysis
by incorporating refinement techniques via path conditions
to improve the precision of the flow control.
Such PDG-based information flow control is more
precise but more expensive than type-based approaches.
However, none of the above works has addressed the issue of flow analysis in virtualised computing systems regarding observations on both executions of communicating processes and affections of cache timing channels.

\section{Conclusions}
\label{sec:conc}
We have proposed a language-based approach for information leakage analysis and control in virtualised computing infrastructure with a concern of cache timing attacks. 
Specifically, 
we have introduced a distributed process algebra with processes able to capture flow security characters in a virtualised environment;
we have described a cache flow policy for leakage analysis through communication covert channels;
and we have presented a type system of the language to enforce the flow policy.
In our future work, we plan to derive concrete implementation of our approach and to extend the current model to define and enforce more practical security policies.

\bibliographystyle{splncs03} 
\bibliography{BIB-cloud}

\begin{thebibliography}{10}
\providecommand{\url}[1]{\texttt{#1}}
\providecommand{\urlprefix}{URL }

\bibitem{AmtoftB04}
Amtoft, T., Banerjee, A.: Information flow analysis in logical form. In: SAS.
  pp. 100--115 (2004)

\bibitem{CapecchiCDR10}
Capecchi, S., Castellani, I., Dezani{-}Ciancaglini, M., Rezk, T.: Session types
  for access and information flow control. In: {CONCUR}. pp. 237--252 (2010)

\bibitem{ClarkHH02}
Clark, D., Hankin, C., Hunt, S.: Information flow for algol-like languages.
  Comput. Lang.  28(1),  3--28 (2002)

\bibitem{CraneHBLF15}
Crane, S., Homescu, A., Brunthaler, S., Larsen, P., Franz, M.: Thwarting cache
  side-channel attacks through dynamic software diversity. In: NDSS (2015)

\bibitem{Denning77}
Denning, D.E., Denning, P.J.: Certification of programs for secure information
  flow. Commun. ACM  20(7),  504--513 (1977)

\bibitem{DomnitserJLAP12}
Domnitser, L., Jaleel, A., Loew, J., Abu{-}Ghazaleh, N.B., Ponomarev, D.:
  Non-monopolizable caches: Low-complexity mitigation of cache side channel
  attacks. {TACO}  8(4),  35:1--35:21 (2012)

\bibitem{GodfreyZ14}
Godfrey, M.M., Zulkernine, M.: Preventing cache-based side-channel attacks in a
  cloud environment. {IEEE} Trans. Cloud Computing  2(4),  395--408 (2014)

\bibitem{GoguenMeseguer82}
Goguen, J., Meseguer, J.: Security policies and security models. In: S \& P.
  pp. 11--20 (1982)

\bibitem{HammerS09}
Hammer, C., Snelting, G.: Flow-sensitive, context-sensitive, and
  object-sensitive information flow control based on program dependence graphs.
  Int. J. Inf. Sec.  8(6),  399--422 (2009)

\bibitem{Hoare78}
Hoare, C.A.R.: Communicating sequential processes. Commun. {ACM}  21(8),
  666--677 (1978)

\bibitem{HondaVY00}
Honda, K., Vasconcelos, V.T., Yoshida, N.: Secure information flow as typed
  process behaviour. In: {ESOP}. pp. 180--199 (2000)

\bibitem{HuntSands06}
Hunt, S., Sands, D.: On flow-sensitive security types. In: POPL. pp. 79--90.
  ACM Press (January 2006)

\bibitem{HuntSands11}
Hunt, S., Sands, D.: From exponential to polynomial-time security typing via
  principal types. In: ESOP. pp. 297--316 (2011)

\bibitem{KimPM12}
Kim, T., Peinado, M., Mainar{-}Ruiz, G.: {STEALTHMEM:} system-level protection
  against cache-based side channel attacks in the cloud. In: USENIX. pp.
  189--204 (2012)

\bibitem{KongASZ13}
Kong, J., Acii{\c{c}}mez, O., Seifert, J., Zhou, H.: Architecting against
  software cache-based side-channel attacks. {IEEE} Trans. Computers  62(7),
  1276--1288 (2013)

\bibitem{LiuHMHTS15}
Liu, C., Harris, A., Maas, M., Hicks, M.W., Tiwari, M., Shi, E.: Ghostrider:
  {A} hardware-software system for memory trace oblivious computation. In:
  ASPLOS. pp. 87--101 (2015)

\bibitem{LiuGYMRHL16}
Liu, F., Ge, Q., Yarom, Y., McKeen, F., Rozas, C.V., Heiser, G., Lee, R.B.:
  Catalyst: Defeating last-level cache side channel attacks in cloud computing.
  In: HPCA. pp. 406--418 (2016)

\bibitem{LiuL14}
Liu, F., Lee, R.B.: Random fill cache architecture. In: MICRO. pp. 203--215
  (2014)

\bibitem{Myers99}
Myers, A.C.: Jflow: Practical mostly-static information flow control. In: POPL.
  pp. 228--241 (1999)

\bibitem{Pottier02}
Pottier, F.: A simple view of type-secure information flow in the p-calculus.
  In: CSFW. pp. 320--330 (2002)

\bibitem{RajNSE09}
Raj, H., Nathuji, R., Singh, A., England, P.: Resource management for isolation
  enhanced cloud services. In: CCSW. pp. 77--84 (2009)

\bibitem{RaneLT15}
Rane, A., Lin, C., Tiwari, M.: Raccoon: Closing digital side-channels through
  obfuscated execution. In: USENIX. pp. 431--446 (2015)

\bibitem{RenFKDD19}
Ren, L., Fletcher, C.W., Kwon, A., van Dijk, M., Devadas, S.: Design and
  implementation of the ascend secure processor. {IEEE} Trans. Dependable Sec.
  Comput.  16(2),  204--216 (2019)

\bibitem{RistenpartTSS09}
Ristenpart, T., Tromer, E., Shacham, H., Savage, S.: Hey, you, get off of my
  cloud: exploring information leakage in third-party compute clouds. In: CCS.
  pp. 199--212 (2009)

\bibitem{RyanS99}
Ryan, P.Y.A., Schneider, S.A.: Process algebra and non-interference. In:
  {CSFW}. pp. 214--227 (1999)

\bibitem{ShiSCZ11}
Shi, J., Song, X., Chen, H., Zang, B.: Limiting cache-based side-channel in
  multi-tenant cloud using dynamic page coloring. In: {IEEE/IFIP} DSN-W. pp.
  194--199 (2011)

\bibitem{TaghdiriSS10}
Taghdiri, M., Snelting, G., Sinz, C.: Information flow analysis via path
  condition refinement. In: FAST. pp. 65--79 (2010)

\bibitem{VolpanoS97}
Volpano, D.M., Smith, G.: A type-based approach to program security. In:
  TAPSOFT. pp. 607--621 (1997)

\bibitem{WangL06}
Wang, Z., Lee, R.B.: Covert and side channels due to processor architecture.
  In: ACSAC. pp. 473--482 (2006)

\bibitem{WangL07}
Wang, Z., Lee, R.B.: New cache designs for thwarting software cache-based side
  channel attacks. In: ISCA. pp. 494--505 (2007)

\bibitem{WangL08}
Wang, Z., Lee, R.B.: A novel cache architecture with enhanced performance and
  security. In: MICRO. pp. 83--93 (2008)

\bibitem{WuXW12}
Wu, Z., Xu, Z., Wang, H.: Whispers in the hyper-space: High-speed covert
  channel attacks in the cloud. In: USENIX. pp. 159--173 (2012)

\bibitem{ZhangJRR12}
Zhang, Y., Juels, A., Reiter, M.K., Ristenpart, T.: Cross-vm side channels and
  their use to extract private keys. In: {CCS}. pp. 305--316 (2012)

\bibitem{ZhouRZ16}
Zhou, Z., Reiter, M.K., Zhang, Y.: A software approach to defeating side
  channels in last-level caches. In: CCS. pp. 871--882 (2016)

\end{thebibliography}

\end{document}